\documentclass[journal,11pt, draftclsnofoot, onecolumn]{IEEEtran}
\IEEEoverridecommandlockouts

\usepackage[cmex10]{mathtools}
\usepackage{mathrsfs}
\usepackage[export]{adjustbox}

\usepackage[noadjust]{cite}
\usepackage{amsmath,amssymb,amsfonts}

\usepackage{mathtools}
\usepackage{algorithmic}
\usepackage{graphicx}
\usepackage{textcomp}
\usepackage{amsthm}
\usepackage{bm}
\usepackage{dsfont}
\usepackage{enumitem}
\usepackage{comment}

\usepackage{color,soul}
\usepackage{multirow,array}
\usepackage{blkarray}
\graphicspath{{./images/}}
\usepackage{amssymb}
\usepackage{amsmath}
\usepackage{tikz}

\usepackage{enumitem}

\newtheorem{theorem}{Theorem}
\newtheorem{lemma}{Lemma}
\newtheorem{remark}{Remark}
\newtheorem{claim}{Claim}

\newtheorem{definition}{Definition}
\newtheorem{innercustomthm}{Theorem}
\newenvironment{customthm}[1]{\renewcommand\theinnercustomthm{#1}\innercustomthm}{\endinnercustomthm}

\usepackage[T1]{fontenc}
\usepackage{tabularx,ragged2e,booktabs}

\usepackage{ifthen}
\newboolean{extend_v}
\newboolean{editor}

\begin{document}

\setboolean{extend_v}{true} 
\setboolean{editor}{true} 

\title{On Pseudolinear Codes for Correcting Adversarial Errors \\
\thanks{This work was supported in part by the U.S National Science Foundation under Grants CCF-1908308, CNS-2128448 and CNS-2212565, and by the Office of Naval Research under grant ONR N000142112472.}
}

\author{\IEEEauthorblockN{Eric Ruzomberka\IEEEauthorrefmark{1}, Homa Nikbakht\IEEEauthorrefmark{1}, Christopher G. Brinton\IEEEauthorrefmark{2}, and H. Vincent Poor\IEEEauthorrefmark{1}}
\IEEEauthorblockA{ \IEEEauthorrefmark{1}\textit{Princeton University \IEEEauthorrefmark{2}Purdue University }}
}

\maketitle

\begin{abstract}

We consider error-correction coding schemes for \textit{adversarial wiretap channels} (AWTCs) in which the channel can a) read a fraction of the codeword bits up to a bound $r$ and b) flip a fraction of the bits up to a bound $p$. The channel can freely choose the locations of the bit reads and bit flips via a process with unbounded computational power. Codes for the AWTC are of broad interest in the area of information security, as they can provide data resiliency in settings where an attacker has limited access to a storage or transmission medium.

We investigate a family of non-linear codes known as \textit{pseudolinear codes}, which were first proposed by Guruswami and Indyk (FOCS 2001) for constructing list-decodable codes independent of the AWTC setting. Unlike general non-linear codes, pseudolinear codes admit efficient encoders and have succinct representations. We focus on unique decoding and show that random pseudolinear codes can achieve rates up to the binary symmetric channel (BSC) capacity $1-H_2(p)$ for any $p,r$ in the \textit{less noisy region}: $p<1/2$ and $r<1-H_2(p)$ where $H_2(\cdot)$ is the binary entropy function. Thus, pseudolinear codes are the first known optimal-rate binary code family for the less noisy AWTC that admit efficient encoders. The above result can be viewed as a derandomization result of random general codes in the AWTC setting, which in turn opens new avenues for applying derandomization techniques to randomized constructions of AWTC codes. Our proof applies a novel concentration inequality for sums of random variables with limited independence which may be of interest as an analysis tool more generally.

\end{abstract}



  \section{Introduction}

  A central problem in coding theory asks how to construct error-correction codes that can correct a large number of errors. Naturally, the solution to this problem is sensitive to the underlying noise model, i.e., the assumptions made about the noise process generating the errors. In this article, we study the problem under an \textit{adversarial} noise model in which the noise process is controlled by a malicious agent or adversary who seeks to generate uncorrectable errors subject to a bound $p \in (0,1)$ on the fraction of codeword symbols in error. We are interested in solutions for unique decoding and where the resulting code constructions are amenable to both efficient encoding and efficient decoding algorithms.

  We study a version of the problem that is of broad interest within the area of information security. Here, the adversary has limited access to the codeword based on access restrictions of the storage or transmission medium. Consider the binary adversarial wiretap channel ($\text{AWTC}_{p,r}$) where the adversary can read a fraction $r \in [0,1]$ of the codeword bits before flipping a fraction $p\in(0,1/2)$ of the bits.\footnote{The $\text{ATWC}_{p,r}$ is known in the literature by other names, including the AWTC type II \cite{Ozarow1986,Aggarwal2009,Wang2016_2} and the limited view adversarial channel \cite{Wang2015}. The myopic adversarial model \cite{Sarwate2010a,Dey2019a} is a general framework for modeling adversaries with partial knowledge of the transmitted codeword, which includes the $\text{ATWC}_{p,r}$ as a special case.} The fraction $r$ parameterizes the adversary's knowledge of the codeword where the special cases $r=0$ and $r=1$ correspond to an adversary that is oblivious and omniscient to the codeword, respectively. We emphasize that no assumptions are made about the adversary's computational complexity.
  
   Codes for the $\text{AWTC}_{p,r}$ are particularly attractive for applications in information security due to their potential for achieving high rates without the need for the complexity assumptions. In fact, finding efficient constructions for AWTC-like settings has been an important open problem in wireless security, as such constructions would enable novel applications like coding-based resilient and confidential communication and authentication \cite{Trappe2015}.  Surprisingly, rates up to the binary symmetric channel (BSC) capacity $1-H_2(p)$ are known to be achievable (by inefficient codes) when the channel from Alice and Bob is \textit{less noisy} than the channel from Alice to the adversary, i.e., when $p \in (0,1/2)$ and $r < 1 -H_2(p)$ where $H_2(\cdot)$ is the binary entropy function \cite{Wang2016,Dey2019a}. However, optimal rate constructions with efficient encoding and decoding algorithms are not known for the less noisy $\text{AWTC}_{p,r}$, except for the special case where $r=0$ \cite{Guruswami2016}.\footnote{We focus on the binary alphabet setting. For larger alphabets, some efficient constructions are known (e.g., see the paper of Wang and Safavi-Naini \cite{Wang2015}).} Moreover, as far as the authors are aware, optimal rate constructions that admit efficient encoders remain unknown. 

   In this article, we investigate a family of codes known as \textit{pseduolinear codes} for error-correction\footnote{In information security, it is often desirable for code construction to satisfy both reliability and secrecy constraints. In this paper, due to space limitations, we address only the reliability question (i.e. error-correction) since it is the harder of the two problems. However, our constructions can be shown to simultaneously achieve both reliability and strong secrecy constraints by combining techniques presented in this paper and those of prior work on secrecy (e.g., \cite{Goldfeld2016}).} on the $\text{AWTC}_{p,r}$. Pseudolinear codes (defined shortly) were first proposed by Guruswami and Indyk \cite{Guruswami2001_2} for list-decoding in the omniscient $\text{AWTC}_{p,r=1}$ setting. Since their introduction, these codes have been studied in a list-decoding setting \cite{Guruswami2002_2,Guruswami2001} but have not been considered for unique decoding in the less noisy $\text{AWTC}_{p,r}$ setting. Pseudolinear codes have a number of nice properties which make them attractive candidates for $\text{AWTC}_{p,r}$ codes, including efficient encodings as well as non-linear properties. We remark that non-linearity is a necessary condition for optimal rate codes for the less noisy $\text{AWTC}_{p,r}$, as linear codes are known to achieve rates strictly less than BSC capacity on the $\text{AWTC}_{p,r}$ \cite{Langberg2008ObliviousCapacity}.\footnote{Such a strict separation between the performance of linear codes and general (linear and non-linear) codes is rare in coding theory. Another coding problem that has this separation is the problem of list-decoding erasures \cite{Guruswami2002_2}. Pseudolinear codes were initially proposed for the problem of list-decoding errors, which at the time of their proposal, addressed a discrepancy between known results on linear codes and general codes for list-decoding errors. Since then, this discrepancy has been resolved \cite{Guruswami2002,Guruswami2010}.}

  \subsection{Pseudolinear Codes} \label{sec:pl_codes}

   For a blocklength $n\geq 1$, rate $R \in (0,1]$ and parameter $k \geq 2$, an $(n,Rn,k)$-pseudolinear code $\mathcal{C} \subseteq \{0,1\}^n$ is defined as follows. Let $H$ be the parity check matrix of some binary linear code of blocklength $2^{Rn}-1$, dimension $2^{Rn}-1-m$ for some $m = O(kRn)$, and minimum distance at least $k+1$. For example, $H$ can be the parity check matrix of a Bose–Chaudhuri–Hocquenghem (BCH) code of design distance $k+1$. The code $\mathcal{C}$ maps a message $u \in \{0,1\}^{Rn}$ to its corresponding codeword in two stages. First, $\mathcal{C}$ performs a non-linear mapping of $u$ to the $u^{\text{th}}$ column of $H$, which we denote as $h^m(u)$.\footnote{To account for the message of all $0$ bits, we define $h^m(0)$ to be the zero vector.} In the second stage, $\mathcal{C}$ performs a linear mapping of $h^m(u)$ to codeword $G h^m(u)$ where $G \in \{0,1\}^{n \times m}$ is the ``generator'' matrix of the code. Hence, the non-linearity is confined to the first stage of the encoding process.

 We highlight three properties of pseudolinear codes which were first presented in \cite{Guruswami2001_2}. First, pseudolinear codes have a succinct representation as $nm = O(k n^2)$ bits are sufficient to describe the generator matrix.\footnote{The description of the parity check matrix $H$ is not included in this accounting since it has an explicit representation (discussed shortly), and thus, $H$ does not need to be stored in memory.} We remark that this number is significantly less than the $n2^{Rn}$ bits required to describe general random codes. Second, encoding is computationally efficient as long as $h^m(u)$ can be efficiently obtained. Indeed, the parity check matrix of a BCH code of blocklength $2^{Rn}-1$, dimension $2^{Rn}-1-m$, and design distance $k+1$ has an explicit representation and can be obtained in time polynomial in $n$ by computing powers of a primitive $(2^{Rn}-1)$th root of unity from the appropriate finite field (see, e.g., \cite{Macwilliams1977}). 
 
 Third, if we consider a \textit{random} pseudolinear code by choosing the generator matrix $G$ at random while fixing the parity check matrix $H$, the codewords of the random code can be shown to be \textit{$k$-wise independent}, i.e., the codewords in every subset of size $k$ are mutually independent. In contrast, codewords of random linear codes are only pairwise (i.e., 2-wise) independent in non-trivial cases. In fact, pseudolinear codes are an extension of one of the standard constructions of Joffe \cite{Joffe1974} of $k$-wise independent binary random variables, which is known to be optimal in terms of its succinct representation up to a constant factor (see, e.g., \cite[Chapter 16]{Alon2008}). The $k$-wise independence property of random pseudolinear codes opens up the possibility of applying the probabilistic method to investigate other properties of these codes in the $\text{AWTC}_{p,r}$ setting. We refer the reader to \cite{Guruswami2001} for further discussion on pseudolinear codes.

  \subsection{Results}

  As our main result, we show that random pseudolinear codes achieve rates arbitrarily close to the BSC capacity for the less noisy $\text{AWTC}_{p,r}$ under the "average error criterion". The average error criterion\footnote{We remark that pseudolinear codes are a \textit{deterministic code} in the sense that the mapping from the message space to the codeword space is deterministic. For deterministic codes, the average error criterion is natural notion of decoding error. Alternatively, one may consider \textit{stochastic codes} in which the above mapping is performed using a random key (private to Alice), and in turn, study the \textit{maximal error criterion} in which decoding may fail over a small subset of keys independent of Alice's message. See \cite{Ahlswede1978} for a comparison of deterministic codes under average error criterion and stochastic codes under maximal error criterion.} allows decoding to fail over a small subset of messages which occurs with some probability that can be made arbitrarily small. The following is an informal version of our main result (Theorem \ref{thm:reliable}). Theorem \ref{thm:reliable} is stated in Section \ref{sec:model}.
  \begin{customthm}{A.}[Informal Version of Theorem \ref{thm:reliable}]
  For $p < 1/2$, $r < 1 - H_2(p)$ and any $\epsilon \in (0,1-H_2(p)-r)$, let rate $R = 1 - H_2(p)- \epsilon$ and let $k$ be at least $\Omega\left(\frac{1+H_2(p)+r + (1/n)\log_2 (1/\delta)}{\epsilon} \right)$. Then with probability at least $1-2^{-\Omega(n)}$ over the design of the generator matrix $G$, a random $(n,Rn,k)$ pseudolinear code $\mathcal{C}$ allows reliable communication over the $\text{AWTC}_{p,r}$ with probability of decoding error $\delta>0$. 
  \end{customthm}
   In Theorem A, we remark that the condition that $k$ be sufficiently large with respect to $(1+H_2(p)+r + (1/n) \log_2(1/\delta))/\epsilon$ is a technical requirement of our proof which is needed to exploit the $k$-wise independence property of $\mathcal{C}$.\footnote{More precisely, it is sufficient for $k$ to be any integer such that $\lfloor k/2 \rfloor > \frac{1+H_2(p)+r + (1/n) \log_2(1/\delta)}{\epsilon/2}$ (c.f. Theorem \ref{thm:reliable}).} It is an open question whether this bound on $k$ is necessary in any sense.

   At the heart of the proof of Theorem A (Theorem \ref{thm:reliable}) is a novel concentration result for sums of random variables with limited independence. The goal of the proof is to bound the probability of decoding error, which is achieved by bounding the number of codewords in the pseudolinear code $\mathcal{C}$ which violate a certain "confusability" condition that can result in a decoding error. While bounding this number, we encounter a notion of limited independence that is strictly weaker than $k$-wise independence. We term this limited independence as \textit{$k$-wise independence over every forest ($k$-wise IOEF)} due to an underlying relationship to acyclic graphs. In Section \ref{sec:conc}, we define $k$-wise IOEF and present a novel concentration inequality for sums of $k$-wise IOEF random variables which is based on the moment method of Schmidt, Siegel, and Srinivasan \cite{Schmidt1995} for sums of $k$-wise independent random variables. An overview and detailed proof of Theorem \ref{thm:reliable} is provided in Section \ref{sec:proof_mr}.

  \subsection{Related Work \& Discussion}

  For the oblivious $\text{AWTC}_{p,r=0}$ setting, general random codes were first shown to achieve BSC capacity by Csisz\'{a}r and Narayan \cite{Csiszar1988TheConstraints} via a complicated but general method-of-types approach. Simpler proofs were later provided in  \cite{Langberg2008ObliviousCapacity,Guruswami2016}. This result was extended to the more general $\text{AWTC}_{p,r}$ setting by Wang \cite{Wang2016} and Dey, Jaggi and Langberg \cite{Dey2019a} under the framework of myopic adversarial channels. We remark that Theorem A can be viewed as a derandomization of known results \cite{Wang2016,Dey2019a} for random general codes on the $\text{AWTC}_{p,r}$.\footnote{The work of \cite{Dey2019a} goes far beyond the $\text{AWTC}_{p,r}$ and provides existential results for a large class of adversarial channel models. For details, see \cite[Table I]{Dey2019a}.}

  Aside from the above existential results, some constructions are known for the $\text{AWTC}_{p,r}$. Guruswami and Smith \cite{Guruswami2016} provide an efficient construction for the oblivious $\text{AWTC}_{p,r=0}$ setting. Their construction provides strong error-correction in the presence of an obvious adversary by obscuring the codeword structure via a random scrambling of the codeword. The decoder undoes the scrambling using control information embedded in the codeword. For the more general less noisy $\text{AWTC}_{p,r}$ setting, one possible construction is to combine a powerful linear code with a non-linear authentication code that can detect additive errors with high probability. This was done in a series of works \cite{Wang2016_2,Safavi-Naini2013,Wang2015} over large alphabets by using modified variants of the algebraic manipulation detection (AMD) codes of Cramer et al. \cite{Cramer2008} as the authentication code in combination with a (Folded) Reed-Solomon Code as an error-correction code. The challenge of extending these ideas to the binary alphabet setting is that no constructions are known for binary codes that can efficiently list-decode errors.

  Lastly, we remark that Theorem A opens new avenues for applying derandomization techniques to random pseudolinear codes to introduce additional useful structure while preserving error-correction power on the $\textit{AWTC}_{p,r}$ with high probability. A recent line of work \cite{Rudra2014,Guruswami2022,Goldberg2023,Ferber2022} has explored derandomization techniques for random linear codes via a random puncturing of known low-rate linear codes. For instance, \cite{Guruswami2022} showed that random puncturings of linear codes with a certain low-bias property are capacity achieving on memoryless additive noise channels. One interesting avenue is to explore whether such derandomization techniques can be adapted to random pseudolinear codes in the $\textit{AWTC}_{p,r}$ setting. Here, one may hope to replace random pseudolinear codes with a derandomized version that is more amendable to tractable decoding algorithms without compromising error-correction power.


  \section{Preliminaries and Formal Results} \label{sec:model}
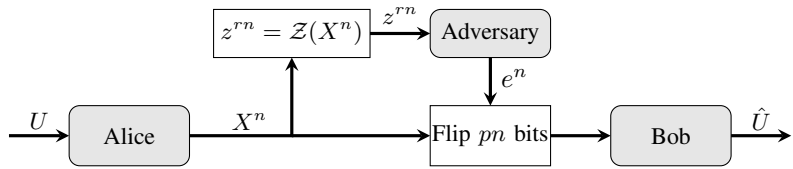
\begin{figure}[t]
  \centering

 
%
  
\begin{tikzpicture}[scale=1.6, >=stealth]
\centering
\footnotesize
\tikzstyle{every node}=[draw,shape=circle, node distance=0.5cm];
 
\draw [ very thick, ->] (-0.5, 0.25)--(0, 0.25);
\node[draw =none] at (-0.25,0.37) {$U$};
\draw [rounded corners, fill = gray!20](0,0) rectangle (1,0.5);
\node[draw =none] at (0.5,0.25) {Alice};
\draw [ very thick, ->] (1, 0.25)--(1.85, 0.25)--(1.85,0.9);
\node[draw =none] at (1.5,0.37) {$X^{n}$};
\draw (1.2,0.9) rectangle (2.5,1.3);
\node[draw =none] at (1.85,1.1) {$z^{rn} = \mathcal{Z}(X^n)$};
\draw [ very thick, ->] (2.5, 1.1)--(3, 1.1);
\node[draw =none] at (2.75,1.25) {$z^{rn}$};
\draw [rounded corners, fill = gray!20] (3,0.9) rectangle (4,1.3);
\node[draw =none] at (3.5,1.1) {Adversary};
\draw [ very thick, ->] (3.5, 0.9)--(3.5, 0.5);
\node[draw =none] at (3.7,0.73) {\small$e^n$};
\draw [ very thick, ->] (1.85, 0.25)--(3, 0.25);
\draw (3,0) rectangle (4,0.5);
\node[draw =none] at (3.5,0.25) {Flip $pn$ bits};
\draw [ very thick, ->] (4, 0.25)--(4.5, 0.25);
\draw [rounded corners, fill = gray!20](4.5,0) rectangle (5.5,0.5);
\node[draw =none] at (5,0.25) {Bob};
\draw [ very thick, ->] (5.5, 0.25)--(6, 0.25);
\node[draw =none] at (5.75,0.39) {$\hat U$};

\end{tikzpicture}
  \caption{The binary adversarial wiretap channel ($\text{AWTC}_{r,p}$). }
   \label{fig:channel_model}
  
  \end{figure}

  \subsection{Notation}
  The Hamming ball of radius $t>0$ centered at a $q$-bit string $x^q \in \{0,1 \}^q$ is the set $\mathcal{B}_{t}(x^q) = \{ y^q \in \{0,1\}^q: d_H\left(y^q,x^q\right) \leq t\}$. Let $\mathds{1}\{\mathcal{H}\}$ denote the indicator of an event $\mathcal{H}$. We let $\mathbb{P}$ and $\mathbb{E}$ denote a probability measure and expectation, respectively, defined w.r.t. some distribution which can either be implied from context or will be stated explicitly. However, when a subscript $\mathcal{C}$ is added as in $\mathbb{P}_{\mathcal{C}}$ and $\mathbb{E}_{\mathcal{C}}$, the distribution is w.r.t. the distribution of a random PL-code $\mathcal{C}$.

  \subsection{Coding Scheme}

  We formally describe our coding scheme for communication over the $\text{AWTC}_{p,r}$. Consider the communication setting of Fig. 
  \ref{fig:channel_model} in which a sender Alice wishes to communicate a message to a receiver Bob over the $\text{AWTC}_{r,p}$. 
  
  \textit{Encoding:} Alice uses a $(n,Rn,k)$ pseudolinear (PL)-code $\mathcal{C}$ as defined in Section \ref{sec:pl_codes}. For a rate $R \in (0,1]$ and blocklength $n \geq 1$, Alice draws a message $U$ uniformly from the message set $\mathcal{U} = \{0,1\}^{Rn}$. She encodes $U$ into an $n$-bit codeword $X^n(U)$ which belongs to code $\mathcal{C}$. Following encoding, Alice transmits $X^n(U)$ over the channel. 
  
  \textit{Decoding:} Bob receives an $n$-bit channel output $X^n(U) \oplus e^n$ where $e^n \in \{0,1\}^n$ is an error string injected by the adversary and the symbol ``$\oplus$'' denotes bitwise XOR. Given the channel output, Bob produces a message estimate $\hat{U} = \phi(X^n(U) \oplus e^n)$ where $\phi$ is the min-distance decoder. We say that a decoding error occurs if $\phi(X^n(U) \oplus e^n) \neq U$.

  \textit{Model of the AWTC:} Let $p \in (0,1/2)$ and $r \in [0,1]$. Before Alice's transmission, the adversary chooses a coordinate set $\mathcal{Z} \in \mathscr{Z}$, where $\mathscr{Z}$ denotes the set of all $rn$-size subsets of $[n]$. Upon transmission, the adversary reads $rn$ bits of Alice's codeword $X^n(U)$ indexed at coordinates $\mathcal{Z}$, which we denote as $\mathcal{Z}(X^n(U))$. In turn, the adversary chooses error string $e^n \in \{0,1\}$ with Hamming weight at most $pn$, i.e., $e^n$ is chosen from the ball $\mathcal{B}_{pn}(0)$. Both choices are made with knowledge of PL-code $\mathcal{C}$.

  In general, the adversary can use \textit{stochastic strategies} for choosing the coordinate set $\mathcal{Z}$ and error string $e^n$. First, the adversary chooses a probability mass function (PMF) $W(\cdot)$ over the domain $\mathscr{Z}$ of coordinate sets, and in turn, draws a coordinate set $\mathcal{Z} \in \mathscr{Z}$ with probability $W(\mathcal{Z})$. Second, the adversary chooses a family of PMFs $\{ W(\cdot|\mathcal{Z},z^{rn}) \}_{\mathcal{Z} \in \mathscr{Z},z^{rn} \in \{0,1\}^{rn}}$ over the domain $\{0,1\}^n$ of error strings, and in turn, upon observing the observation string $z^{rn} \in \{0,1\}^{rn}$, draws the error string $e^n \in \{0,1\}^{n}$ with probability $W(e^n|\mathcal{Z},z^{rn})$. The adversary's bit-flip constraint requires that $W(e^n|\mathcal{Z},z^{rn}) = 0$ for all $e^n \not \in \mathcal{B}_{pn}(0)$. Again, all distributions can depend on PL-code $\mathcal{C}$.

  \textit{Probability of Decoding Error:} Under the above channel model, the probability of decoding error is $P_{\mathrm{error}}(\mathcal{C}) = \mathbb{P}\left( U \neq \phi\left( X^n(U) \oplus e^n \right) \right)$ where $\mathbb{P}$ denotes the probability with respect to (w.r.t) the joint distribution of $U$, $\mathcal{Z}$ and $e^n$. For integer $k\geq 2$, an $(n,Rn,k)$ PL-code $\mathcal{C}$ is said to allow reliable communication over the $\text{AWTC}_{p,r}$ with probability of decoding error $\delta \in (0,1)$ if $P_{\mathrm{error}}(\mathcal{C}) \leq \delta$. 


  \begin{definition}
  Let $F(n,Rn,k)$ be the distribution over $(n,Rn,k)$ PL-codes in which the generator matrix $G$ is chosen uniformly from $\{0,1\}^{n \times m}$ for a fixed parity check matrix $H$.
  \end{definition}

  \begin{definition}[$k$-wise independence]
   A set of random variables $V_1,\ldots,V_M$ are $k$-wise independent for an integer $k \geq 2$ if for any subset $\mathcal{T} \subseteq [M]$ of size $k$ the random variables $\{V_i\}_{i \in \mathcal{\mathcal{T}}}$ are mutually independent.
   \end{definition}

  \begin{lemma} \label{thm:k_wise_cw}
  The codewords of an $(n,Rn,k)$ PL-code $\mathcal{C}$ drawn from the distribution $F(n,Rn,k)$ are uniformly distributed in $\{0,1\}^n$ and $k$-wise independent.
  \end{lemma}

   Proof of Lemma \ref{thm:k_wise_cw} is given in \cite[{Lemma~9.1}]{Guruswami2001}. The proof uses the fact that for any linear code of minimum distance at least $k+1$, any $k$ columns of the parity check matrix $H$ are linearly independent. For completeness, we restate the proof in Appendix \ref{sec:k_wise_cw_proof}. The formal version of Theorem A is as follows.

  \begin{theorem} \label{thm:reliable}
   For $p \in [0,1/2)$, $r \in [0,1-H_2(p))$ and any $\epsilon \in (0,1-H_2(p)-r)$, let the rate $R = 1 - H_2(p) - \epsilon$. With probability at least $$1-\frac{c_k (10n)^{1 + \lfloor k/2 \rfloor} }{\delta}2^{(1+H_2(p)+r-\lfloor k/2 \rfloor \frac{\epsilon}{2})n}$$ over the code distribution $F(n,Rn,k)$ where $c_k$ is some constant that depends only on $k$, an $(n,Rn,k)$ PL-code allows reliable communication over the $\text{AWTC}_{p,r}$ with probability of decoding error $\delta>0$.
  \end{theorem}

  Theorem \ref{thm:reliable} has the following asymptotic interpretation as $n$ tends to $\infty$: For any choice of parameters $p,r,\epsilon,R$ that satisfy the condition of Theorem \ref{thm:reliable} and any fixed $k$ that satisfies the condition $\lfloor k/2 \rfloor (\epsilon/2) > 1 + H_2(p)+r$, an $(n,Rn,k)$ PL-code picked at random will with high probability allow reliable communication over the $\text{AWTC}_{p,r}$ with probability of decoding error $o(1)$.

   \subsection{Concentration Inequalities} \label{sec:conc}

   We present a few concentration inequalities that will be used in the proof of Theorem \ref{thm:reliable}. These results apply to sums of random variables with limited independence. We first restate a concentration inequality of Schmidt, Siegel, and Srinivasan \cite{Schmidt1995} for sums of $k$-wise independent random variables.

  \begin{lemma}[{\cite{Schmidt1995}}] \label{thm:conc_ineq}
  Suppose that $V_1,\ldots, V_M$ are $k$-wise independent binary random variables that take values in $\{0,1\}$ and define $V \triangleq \sum_{i=1}^M V_i$ and $\mu \triangleq \mathbb{E}[V]$. For any $\gamma > 0$,
  \begin{equation} \nonumber
  \mathbb{P}\left( V > \mu(1+\gamma) \right)  \leq \frac{{M \choose k} (\frac{\mu}{M})^k}{{\mu(1+\gamma) \choose k}}.
  \end{equation}

  \end{lemma}

  The next result pertains to sums of random variables with a certain weaker notion of $k$-wise independence, the definition of which requires the following graph theoretic concepts. Let $\mathcal{V}_1$ and $\mathcal{V}_2$ be disjoint sets of integers of size $M_1$ and $M_2$, respectively. Let $\mathcal{G} = (\mathcal{V}_1,\mathcal{V}_2,\mathcal{E})$ denote an undirected bipartite graph with vertex parts $\mathcal{V}_1$ and $\mathcal{V}_2$ and edge set $\mathcal{E}$. By bipartite, we mean that every edge $(i,j) \in \mathcal{E}$ has corresponding vertices $i \in \mathcal{V}_1$ and $j \in \mathcal{V}_2$. An undirected bipartite forest $\mathcal{F}$ is an undirected bipartite graph that is acyclic. We now define a weaker notion of $k$-wise independence for random variables indexed over the product set $\mathcal{V}_1 \times \mathcal{V}_2$.
  
  \begin{definition}[$k$-wise independence over every forest]
  A set of random variables $\{V_{i,j}\}_{(i,j) \in \mathcal{V}_1 \times \mathcal{V}_2}$ is $k$-wise independent over every forest for an integer $k \geq 2$ if for any undirected bipartite forest $\mathcal{F} = (\mathcal{V}_1,\mathcal{V}_2,\mathcal{E})$ with edge set $\mathcal{E}$ of size $k$ the random variables $\{V_{i,j}\}_{(i,j) \in \mathcal{E}}$ are mutually independent.
  \end{definition}

  As far as the authors are aware, the notion of $k$-wise independence over every forest has not been previously investigated. We remark that $k$-wise independence over every forest is related to graphs are used to describe dependencies between random variables, .e.g., [?]. Although strictly weaker than $k$-wise independence, a result similar to Lemma \ref{thm:conc_ineq} holds when $k$-wise independence is relaxed to $k$-wise independence over every forest.

  \begin{lemma} \label{thm:conc_ineq_2}
  Suppose that $\{V_{i,j}\}_{(i,j) \in \mathcal{V}_1 \times \mathcal{V}_2}$ are $k$-wise independent over every forest and binary random variables that take values in $\{0,1\}$, and define the sum $V \triangleq \sum_{i\in \mathcal{V}_1} \sum_{j\in\mathcal{V}_2} V_{i,j}$ and $\mu \triangleq \mathbb{E}[V]$. Furthermore, suppose that $\mu \geq 1$. For any $\gamma > 0$,
  \begin{equation} \nonumber
  \mathbb{P}(V > \mu(1+\gamma)) \leq c_k \frac{{M_1 M_2 \choose {k}} (\frac{\mu}{M_1 M_2})^{k}}{{\mu(1+\gamma) \choose {k}}}
  \end{equation}
  where $c_k$ is a positive constant the depends only on $k$, and where $M_1 = |\mathcal{V}_1|$ and $M_2 = |\mathcal{V}_2|$.
  \end{lemma}

  The proof of Lemma \ref{thm:conc_ineq_2} studies the joint distribution of $\{V_{i,j}\}_{(i,j) \in \mathcal{E}}$ for all subsets $\mathcal{E} \subseteq \mathcal{V}_1 \times \mathcal{V}_2$ of size $k$, where the subset $\mathcal{E}$ can be viewed as an edge set of the undirected bipartite graph $\mathcal{G}(\mathcal{E}) = (\mathcal{V}_1,\mathcal{V}_2,\mathcal{E})$. The proof considers two cases: where $\mathcal{G}(\mathcal{E})$ is acyclic and when $\mathcal{G}(\mathcal{E})$ contains a cycle. In the former case, the random variables $\{V_{i,j}\}_{(i,j) \in \mathcal{E}}$ are independent following the definition of $k$-wise independent over every forest and the approach of Schmidt, Siegel, and Srinivasan \cite{Schmidt1995} can be directly applied. In the latter case, $\{V_{i,j}\}_{(i,j) \in \mathcal{E}}$ are not independent and additional work is required to apply the approach of \cite{Schmidt1995} which involves bounding the number of undirected bipartite graphs that contain at least one cycle. The detailed proof can be found in Appendix \ref{sec:conc_ineq_2_proof}.

  \section{Proof of Theorem \ref{thm:reliable}} \label{sec:proof_mr}

  In this section, we present our proof of Theorem \ref{thm:reliable}. We use the notation and coding scheme described in in Section \ref{sec:model} for communication setting over the $\text{AWTC}_{p,r}$. In the sequel, we fix the fraction $p \in (0,1/2)$ and the fraction $r \in (0,1-H_2(p))$ of Alice's codeword bits that can be read and flipped, respectively, by the $\text{AWTC}_{p,r}$, and set the blocklength $n \geq 1$. For any $\epsilon \in (0,1-H_2(p)-r)$ we set the rate $R = 1 - H_2(p)- \epsilon$. Notice that this choice of $\epsilon$ ensures that $R > r$. In the following, we first provide an overview of the proof followed by a detailed proof.
  

  \subsection{Overview of Proof of Theorem \ref{thm:reliable}}

  Our proof of Theorem \ref{thm:reliable} relates the probability of decoding error $P_{\mathrm{error}}$ to a notion of confusable codewords. Let $e^n$ be an error string in $\{0,1\}^n$ with Hamming weight at most $pn$. We say that a codeword $X^n(u)$ corresponding to a message $u \in \mathcal{U}$ is \textit{confused by the error string $e^n$} if there exists another message $u' \in \mathcal{U} \setminus \{u\}$ such that $X^n(u')$ is contained in the Hamming ball $\mathcal{B}_{pn}(X^n(u) \oplus e^n)$. Let $\mathcal{A}(e^n)$ denote the set of all messages $u \in \mathcal{U}$ such that codeword $X^n(u)$ is confused by $e^n$. Furthermore, for a coordinate set $\mathcal{Z} \in \mathscr{Z}$ and an observation string $z^{rn} \in \{0,1 \}^{rn}$, define the set of messages \textit{consistent with $\mathcal{Z}$ and $z^{rn}$} as $
  \mathcal{O}(\mathcal{Z},z^{rn}) = \left\{ u \in \mathcal{U}: z^n = \mathcal{Z}(X^n(u)) \right\}$. Thus, for a coordinate set $\mathcal{Z}$, an observation of the string $z^{rn}$ informs the adversary that Alice's message belongs to the set $\mathcal{O}(\mathcal{Z},z^{rn})$. Note that both $\mathcal{A}(e^n)$ and $\mathcal{O}(\mathcal{Z},z^{rn})$ depend on Alice's PL-code $\mathcal{C}$, although we have not denoted this explicitly. 
  
  We remark that the intersection of the sets $\mathcal{A}(e^n)$ and $\mathcal{O}(\mathcal{Z},z^{rn})$ indicates all messages which are likely to be Alice's transmitted message (from the adversary's perspective) and which are confused by an error string $e^n$. We show in the following, roughly speaking, that if for an $(n,Rn,k)$ PL-code $\mathcal{C}$ the aforementioned intersection contains a small number of messages for every $e^n \in \mathcal{B}_{pn}(0)$, $\mathcal{Z} \in \mathscr{Z}$ and $z^{rn} \in \{0,1\}^{rn}$, then the probability of decoding error is small.\footnote{We note that the pseudolinear structure of $\mathcal{C}$ is not necessary for this claim.}

  \begin{lemma}[Sufficient Condition 1] \label{thm:suff_cond_1}
  An $(n,Rn,k)$ PL-code $\mathcal{C}$ allows reliable communication over the $\text{AWTC}_{p,r}$ with probability of decoding error at most $\delta>0$ if for every coordinate set $\mathcal{Z} \in \mathscr{Z}$, every observation string $z^{rn} \in \{0,1\}^{rn}$ and every error string $e^n \in \mathcal{B}_{pn}(0)$, $|\mathcal{A}(e^n) \cap \mathcal{O}(\mathcal{Z},z^{rn})| \leq \delta 2^{(R-r)n}.$ We delay the proof until Section \ref{sec:detailedproof}.
  \end{lemma}

    Our approach to prove Theorem \ref{thm:reliable} is to use the probabilistic method in which the PL-code $\mathcal{C}$ is randomly drawn from the distribution $F(n,Rn,k)$ and the set size $|\mathcal{A}(e^n) \cap \mathcal{O}(\mathcal{Z},z^{rn})|$ is viewed as a random variable. Using the sufficient condition of Lemma \ref{thm:suff_cond_1}, it is sufficient to show that the set size $|\mathcal{A}(e^n) \cap \mathcal{O}(\mathcal{Z},z^{rn})|$ is sufficiently small ($\leq \delta 2^{(R-r)n}$) for all $e^n \in \mathcal{B}_{pn}(0)$, $\mathcal{Z} \in \mathscr{Z}$ and $z^{rn} \in \{0,1\}^{rn}$ with high probability over the distribution $F(n,Rn,k)$. One can show this directly by showing that $|\mathcal{A}(e^n) \cap \mathcal{O}(\mathcal{Z},z^{rn})|$ has an expectation (w.r.t. to $F(n,Rn,k)$) that is sufficiently smaller than $\delta 2^{(R-r)n}$, and subsequently, showing that $|\mathcal{A}(e^n) \cap \mathcal{O}(\mathcal{Z},z^{rn})|$ is \textit{concentrated above its expectation}, i.e., realizes a value that is not much larger than its expectation with high probability. In turn, if for any $\mathcal{Z} \in \mathscr{Z}$, $z^{rn} \in \{0,1\}^{rn}$ and $e^n \in \mathcal{B}_{pn}(0)$, the set size $|\mathcal{A}(e^n)\cap \mathcal{O}(\mathcal{Z},z^{rn})|$ is concentrated above its expectation with high enough probability, then a simple union bound can be applied to show that $|\mathcal{A}(e^n) \cap \mathcal{O}(\mathcal{Z},z^{rn})|$ is concentrated above its expectation for \textit{all} $\mathcal{Z} \in \mathscr{Z}$, $z^{rn} \in \{0,1 \}^{rn}$ and $e^n \in \mathcal{B}_{pn}(0)$.
  


  Towards the goal of showing that $|\mathcal{A}(e^n) \cap \mathcal{O}(\mathcal{Z},z^{rn})|$ is concentrated above its expectation, we remark that $|\mathcal{A}(e^n) \cap \mathcal{O}(\mathcal{Z},z^{rn})|$ does not have a form that is directly amenable to our Lemma \ref{thm:conc_ineq_2} concentration inequality. To side-step this issue, we study a quantity $V(\mathcal{Z},z^{rn},e^n)$ related to $|\mathcal{A}(e^n) \cap \mathcal{O}(\mathcal{Z},z^{rn})|$ which has the following three properties for every $\mathcal{Z} \in \mathscr{Z}$, $z^{rn} \in \{0,1\}^{rn}$ and $e^n \in \mathcal{B}_{pn}(0)$:
  \begin{enumerate}[label=\textit{Property \arabic*)},itemindent=*]
  \item  The set size $|\mathcal{A}(e^n) \cap \mathcal{O}(\mathcal{Z},z^{rn})|$ is bounded above by $V(\mathcal{Z},z^{rn},e^n)$.
  \item $V(\mathcal{Z},z^{rn},e^n)$ has an expectation (w.r.t. the distribution $F(n,Rn,k)$) that is sufficiently smaller than $\delta 2^{(R-r)n}$. 
  \item $V(\mathcal{Z},z^{rn},e^n)$ has a structure of limited independence which permits application of the concentration inequality of Lemma \ref{thm:conc_ineq_2}. In particular, for some $k' = k'(k)$, $V(\mathcal{Z},z^{rn},e^n)$ is the sum of random variables that are $k'$-wise independent over every forest.
  \end{enumerate}
  Given a quantity $V(\mathcal{Z},z^{rn},e^n)$ with the above three properties for all $\mathcal{Z} \in \mathscr{Z}$, $z^{rn} \in \{0,1\}^{rn}$ and $e^n \in \mathcal{B}_{pn}(0)$, one can apply Lemma \ref{thm:conc_ineq_2} to show that $V(\mathcal{Z},z^{rn},e^n)$ is concentrated above its expectation, and in turn, imply that $|\mathcal{A}(e^n) \cap \mathcal{O}(\mathscr{Z},z^{rn})|$ is concentrated above its expectation.

  The specific quantity $V(\mathcal{Z},z^{rn},e^n)$ we consider is the following. For an integer $S$ large enough such that the following subsets exist, let $\mathcal{S}_1, \ldots, \mathcal{S}_S$ be subsets of the message set $\mathcal{U}$ with the property that for any two unique messages $u,u' \in \mathcal{U}$ there exists an index $i \in [S]$ such that $u \in \mathcal{S}_i$ and $u' \in \mathcal{S}^c_i$. In turn, define the sum
  \begin{equation} \label{eq:V}
  V(\mathcal{Z},z^{rn},e^n) \triangleq \sum_{i=1}^S \sum_{u \in \mathcal{O}(\mathcal{Z},z^{rn}) \cap \mathcal{S}_i} \sum_{u' \in \mathcal{S}^c_i} V(u,u',e^n)
  \end{equation}
  where $V(u,u',e^n) \triangleq \mathds{1} \{ X^n(u') \in \mathcal{B}_{pn}\left( X^n(u) \oplus e^n \right) \}$. We can immediately verify that Property 1 holds for the above choice of $V(\mathcal{Z},z^{rn},e^n)$ via the following inequalities:
  \begin{align}
  &|\mathcal{A}(e^n) \cap \mathcal{O}(\mathcal{Z},z^{rn})| = \sum_{u \in \mathcal{O}(\mathcal{Z},z^{rn})} \hspace{-1em}\mathds{1}\{\exists u' \in \mathcal{U} \setminus \{u\} \text{ s.t. } X^n(u') \in \mathcal{B}_{pn}(X^n(u) \oplus e^n)\} \nonumber \\
  & \leq \sum_{u \in \mathcal{O}(\mathcal{Z},z^{rn})} \sum_{u' \in \mathcal{U} \setminus \{u\}} V(u,u',e^n) \stackrel{(a)}{\leq} V(\mathcal{Z},z^{rn},e^n) \nonumber
  \end{align}
  where $(a)$ follows from the definition of the subsets $\mathcal{S}_1, \ldots, \mathcal{S}_S$. Furthermore, it is straightforward to verify (and will be verified later in the proof) that Property 2 holds for sufficiently small $S$. However, to show that Property 3 holds, additional work is needed. In particular, we must show that for every $\mathcal{Z} \in \mathscr{Z}$, $z^{rn} \in \{0,1\}^{rn}$ and $e^n \in \mathcal{B}_{pn}(0)$ the random variables $\{ V(u,u',e^n) : (u,u') \in \mathcal{O}(\mathcal{Z},z^{rn}) \cap \mathcal{S}_i \times \mathcal{S}^c_i \}$ are $\lfloor k/2  \rfloor$-wise independent over every forest by using the property that the codewords of $\mathcal{C}$ are $k$-wise independent. This is difficult as the structure of the set $\mathcal{O}(\mathcal{Z},z^{rn})$ is random, and thus, the definition of $k$-wise independence over every forest cannot be directly applied. We remark that naively conditioning on the messages in $\mathcal{O}(\mathcal{Z},z^{rn})$ is not a viable approach, as conditioning can invalidate the property of $k$-wise independent codewords.\footnote{We remark that when codewords are fully independent, conditioning on the messages in $\mathcal{O}(\mathcal{Z},z^{rn})$ is a viable approach, as demonstrated in \cite{Dey2019a}.} Thus, the focus of the remainder of the proof is to show Property 3.
  
  In the following overview, we outline the remainder of the proof of Theorem \ref{thm:reliable}.
  \begin{itemize}
  \item For $S = O(n)$, we show that there exists $S$ subsets $\mathcal{S}_1, \ldots, \mathcal{S}_S$ of $\mathcal{U}$ such that for any two unique messages $u,u' \in \mathcal{U}$ there exists an $i \in [S]$ such that $u \in \mathcal{S}_i$ and $u' \in \mathcal{S}_i$ (c.f. Lemma \ref{thm:S_set_constr}). In the sequel, we fix both $S= O(n)$ and the sets $\mathcal{S}_1, \ldots, \mathcal{S}_S$.
  \item Recall that the following two descriptions are equivalent: the PL-code $\mathcal{C}$ is drawn from the distribution $F(n,Rn,k)$ and the generator matrix $G$ of $\mathcal{C}$ is uniformly distributed in $\{0,1\}^{n \times m}$. We introduce a partition of the generator matrix $G$ based on the adversary's coordinates $\mathcal{Z}$. Let $\mathcal{Z}(G)$ denote the binary $rn \times m$ matrix formed by the $rn$ rows of $G$ indexed by the coordinates $\mathcal{Z}$. Similarly, let $\mathcal{Z}^c(G)$ denote the binary $(1-r)n \times m$ matrix formed by the $(1-r)n$ rows of $G$ indexed by the coordinates $\mathcal{Z}^c$. It follows from the uniform distribution of $G$ that $\mathcal{Z}(G)$ and $\mathcal{Z}^c(G)$ are uniformly distributed and independent.
  \item Observe that for any message $u \in \mathcal{U}$, the codeword views $\mathcal{Z}(X^n(u))$ and $\mathcal{Z}^c(X^n(u))$ depend deterministically on $\mathcal{Z}(G)$ and $\mathcal{Z}^c(G)$, respectively, and are independent of $\mathcal{Z}^c(G)$ and $\mathcal{Z}(G)$, respectively. 
  \item Following the above observation, the set $\mathcal{O}(\mathcal{Z},z^{rn})$ depends only on the partition $\mathcal{Z}(G)$ and not on $\mathcal{Z}^c(G)$. Hence, for a fixed $\mathcal{Z} \in \mathscr{Z}$, we can condition on the event $\mathcal{Z}(G) = A$ for any matrix $A \in \{0,1\}^{rn \times m}$ to simultaneously fix the set $\mathcal{O}(\mathcal{Z},z^{rn})$ for all $z^{rn} \in \{0,1\}^{rn}$ while preserving the $k$-wise independence of the codeword views $\{\mathcal{Z}^c(X^n(u))\}_{u \in \mathcal{U}}$ (c.f. Lemma \ref{thm:E_comp_2}). We remark that this step highlights the useful structure of pseudolinear codes beyond their property of $k$-wise independent codewords.
  \item  \begin{definition} \label{def:H}
  \vspace{-1.8em} For a coordinate set $\mathcal{Z} \in \mathscr{Z}$ and some parameter $\theta>0$, define $\mathcal{H}(\mathcal{Z})$ to be the event that there exists some observation string $z^{rn} \in \{0,1\}^{rn}$ such that $|\mathcal{O}(\mathcal{Z},z^{rn})| > 2^{(R-r)n + \theta n}$. 
  \end{definition}
  \item   \begin{lemma} \label{thm:E_comp_1}
  \vspace{-1.8em} For a fixed $\mathcal{Z} \in \mathscr{Z}$, event $\mathcal{H}(\mathcal{Z})$ occurs with probability at most $\alpha_k 2^{-n(k \theta-r)}$ over distribution $F(n,Rn,k)$ where $\alpha_k$ is a constant that depends only on $k$.
  \end{lemma}
  The proof of Lemma \ref{thm:E_comp_1} follows standard typicality arguments. For example, one can apply the concentration inequality of Lemma \ref{thm:conc_ineq} for sums of $k$-wise independent random variables.
  \item   \begin{definition}
  \vspace{-1.8em} A matrix $A \in \{0,1\}^{rn \times n}$ is said to be consistent with $\mathcal{H}^c(\mathcal{Z})$ if $\mathbb{P}_{\mathcal{C}}\left( \mathcal{Z}(G)=A \big| \mathcal{H}(\mathcal{Z})\right)>0$ where $\mathbb{P}_{\mathcal{C}}$ denotes the probability w.r.t. $\mathcal{C} \sim F(n,Rn,k)$.
  \end{definition}
  \item \textit{(Property 2)} For any matrix $A$ consistent with $\mathcal{H}(\mathcal{Z})$, conditional on $\mathcal{H}^c(\mathcal{Z})$ and $\mathcal{Z}(G) = A$, the expectation of $V(\mathcal{Z},z^{rn},e^n)$ is smaller than $\delta 2^{(R-r)n}$ by a factor exponential in $n$ for small enough $\theta$ (c.f. Lemma \ref{thm:cond_Vi}).
  \item \textit{(Property 3)} For a matrix $A$ consistent with $\mathcal{H}(\mathcal{Z})$ and for $i \in [S]$, conditional on $\mathcal{H}^c(\mathcal{Z})$ and $\mathcal{Z}(G) = A$ the random variables $\{V(u,u',e^n): (u,u') \in \mathcal{O}(\mathcal{Z},z^{rn}) \cap \mathcal{S}_i \times \mathcal{S}_i^c \}$ are $\lfloor k/2 \rfloor$-wise independent over every forest (c.f. Lemma \ref{thm:Vk_wise_forest}).
  \item We apply the concentration inequality of Lemma \ref{thm:conc_ineq_2}. For large enough $k$ and with probability at most $2^{-k\Omega(n)}$ over the distribution $F(n,Rn,k)$, $V(\mathcal{Z},z^{rn},e^n)$ is bounded above by $\delta 2^{(R-r)n}$ for every $\mathcal{Z} \in \mathscr{Z}$, $z^{rn} \in \{0,1\}^{rn}$ and $e^n \in \mathcal{B}_{pn}(0)$ (c.f. Lemma \ref{thm:suff2} and Lemma \ref{thm:V_conc}).
  
  \end{itemize}

  \subsection{Detailed Proof of Theorem \ref{thm:reliable}} \label{sec:detailedproof}

  We begin by proving Lemma \ref{thm:suff_cond_1}.
  \begin{proof}[Proof of Lemma \ref{thm:suff_cond_1}]
  The proof is similar in spirit to the proof of \cite[Lemma 2.2]{Langberg2008ObliviousCapacity}. Let $\delta > 0$ and suppose that $\mathcal{C}$ is an $(n,Rn,k)$ PL-code such that for every $\mathcal{Z} \in \mathscr{Z}$, every $z^{rn} \in \{0,1\}^{rn}$ and every $e^n \in \mathcal{B}_{pn}(0)$, $|\mathcal{A}(e^n) \cap \mathcal{O}(\mathcal{Z},z^{rn})| \leq \delta 2^{(R-r)n}$.

 Given that Alice uses the PL-code $\mathcal{C}$, the probability of decoding error $P_{\mathrm{error}}(\mathcal{C})$ is by definition equal to the probability that the decoded message $\phi(X^n(U) \oplus e^n)$ is not equal to Alice's message $U$ where $\phi$ is the min-distance decoder. This probability is bounded above by the probability that the message $U$ is confused by $e^n$, i.e., 
  \begin{align} 
   \mathbb{P}\left( U \in \mathcal{A}(e^n) \right)  &= \sum_{u \in \mathcal{U}} \mathbb{P}\left(u \in \mathcal{A}(e^n) | U=u \right) 2^{-Rn} = \sum_{e^n \in \mathcal{B}_{pn}(0)} \sum_{u \in \mathcal{A}(e^n) } W(e^n|u) 2^{-Rn} \label{eq:suff1_1}
  \end{align}
  where $\mathbb{P}$ is the probability w.r.t. the joint distribution of $U$, $\mathcal{Z}$ and $e^n$. Using the fact that $e^n$ is conditionally independent of $U$ given the adversary's observation string $z^{rn}$ and coordinate set $\mathcal{Z}$, $W(e^n|u)$ is equal to $$\sum_{\mathcal{Z} \in \mathscr{Z}} \sum_{z^{rn} \in \{0,1\}^{rn}} W(e^n|\mathcal{Z},z^{rn})W(\mathcal{Z}) \mathds{1}\{ u \in \mathcal{O}(\mathcal{Z},z^{rn}) \},$$ and in turn, (\ref{eq:suff1_1}) is equal to 
  \begin{align}
  & \sum_{e^n \in \mathcal{B}_{pn}(0)} \sum_{\mathcal{Z} \in \mathscr{Z}} \sum_{z^{rn} \in \{0,1\}^{rn}} \sum_{u \in \mathcal{A}(e^n) \cap \mathcal{O}(\mathcal{Z},z^{rn})} W(e^n|z^{rn}) W(\mathcal{Z}) 2^{-Rn}  \nonumber \\
  & = \sum_{e^n \in \mathcal{B}_{pn}(0)} \sum_{\mathcal{Z} \in \mathscr{Z}} \sum_{z^{rn} \in \{0,1\}^{rn}} W(e^n|z^{rn}) W(\mathcal{Z})2^{-Rn}|\mathcal{A}(e^n) \cap \mathcal{O}(\mathcal{Z},z^{rn})|  \stackrel{(a)}{\leq} \delta \nonumber
  \end{align}
  where (a) follows from the supposition that $|\mathcal{A}(e^n) \cap \mathcal{O}(\mathcal{Z},z^{rn})| \leq \delta 2^{(R-r)n}$ for all $\mathcal{Z} \in \mathscr{Z}$, $z^{rn} \in \{0,1\}^{rn}$ and $e^n \in \mathcal{B}_{pn}(0)$.
  \end{proof}

  \begin{lemma} \label{thm:S_set_constr}
  There exists a collection $\{ \mathcal{S}_i \}_{i=1}^{10n}$ of $S = 10 n$ subsets of $\mathcal{U}$ such that for any two unique messages $u,u'$ in $\mathcal{U}$, there exists an index $i^* \in [S]$ such that both $u \in \mathcal{S}_{i^*}$ and $u' \in \mathcal{S}^c_{i^*}$.
  \end{lemma}

  \begin{proof}[Proof of Lemma \ref{thm:S_set_constr}]
  We prove Lemma \ref{thm:S_set_constr} via the probabilistic method. Let $S$ be a positive integer and construct the subsets $\mathcal{S}_1, \ldots, \mathcal{S}_S$ randomly and independently. For $i=1,\ldots,S$, subset $\mathcal{S}_i$ is constructed via the following: for every message $u \in \mathcal{U}$, $u$ belongs to $\mathcal{S}_i$ with probability $1/2$, and $u$ belongs to $\mathcal{S}^c_i$ with probability $1/2$.

  We bound the probability over the subset constructions that there exists two unique messages $u,u' \in \mathcal{U}$ such that $u \not\in \mathcal{S}_i$ or $u' \not\in \mathcal{S}_i^c$ for all $i \in [S]$. By a union bound this probability is bounded above by 
  \begin{align}
  &2^{2Rn} \mathbb{P}\left(\cap_{i=1}^S \{u \not\in \mathcal{S}_i \text{ or } u' \not\in \mathcal{S}^c_i\}\right) = 2^{2Rn} (3/4)^{S} \label{eq:S_set_constr_1}
  \end{align}
  where $\mathbb{P}(\cdot)$ denotes the probability over the subset construction and where (\ref{eq:S_set_constr_1}) follows from the property that $\mathcal{S}_1, \ldots, \mathcal{S}_S$ are constructed independently of each other. Setting $S = 10n$, we ensure that (\ref{eq:S_set_constr_1}) is strictly less than $1$ for any $R \leq 1$, and thus, there exists a collection of subsets as described in Lemma \ref{thm:S_set_constr}.
  \end{proof}

  In the sequel, set $S = 10n$ and let $S_1, \ldots, \mathcal{S}_S$ be some subsets identified in Lemma \ref{thm:S_set_constr}. Recall that $V(\mathcal{Z},z^{rn},e^n)$ as defined in (\ref{eq:V}) depends on this choice of sets. We now present a second sufficient condition for the proof of Theorem \ref{thm:reliable} that makes use of the structure of $\mathcal{H}(\mathcal{Z})$.

  \begin{lemma}[Sufficient Condition 2] \label{thm:suff2}
  Suppose that for every coordinate set $\mathcal{Z} \in \mathscr{Z}$, every observation string $z^{rn} \in \{0,1\}^{rn}$, every error string $e^n \in \mathcal{B}_{pn}(0)$ and every matrix $A \in \{0,1\}^{rn \times n}$ consistent with $\mathcal{H}^c(\mathcal{Z})$,
  \begin{equation}
   \mathbb{P}_{\mathcal{C}}\left(V(\mathcal{Z},z^{rn},e^n) > \delta 2^{(R-r)n} \big| \mathcal{H}^c(\mathcal{Z}),\mathcal{Z}(G)=A \right) \leq \beta(n)
  \end{equation}
  for some $\beta(n)>0$. Then with probability at least $1 - 2^{(1+r+H_2(p))n}\beta(n) - \alpha_k 2^{(1+r-k\theta)}$ over the code design, the $(n,Rn,k)$ PL-code $\mathcal{C}$ allows reliable communication over the $\text{AWTC}_{p,r}$ with probability of decoding error $\delta>0$.
  \end{lemma}

  \begin{proof}[Proof of Lemma \ref{thm:suff2}]
  From our first sufficient condition (Lemma \ref{thm:suff_cond_1}), a randomly chosen code from the distribution $F(n,Rn,k)$ \textit{does not} allow reliable communication over the $\text{AWTC}_{p,r}$ with probability of decoding error $\delta>0$ with probability at most
  \begin{align} \label{eq:suff2_1}
  \mathbb{P}_{\mathcal{C}}\left(\bigcup_{\mathcal{Z} \in \mathscr{Z}} \bigcup_{z^{rn} \in \{0,1\}^{rn}} \bigcup_{e^n \in \mathcal{B}_{pn}(0)} \left\{ V(\mathcal{Z},z^{rn},e^n)> \delta 2^{(R-r)n} \right\} \right)
  \end{align}
  over the code design. Following a union bound and an application of Baye's formula, (\ref{eq:suff2_1}) is bounded above by
  \begin{align} \label{eq:suff2_2}
  \sum_{\mathcal{Z} \in \mathscr{Z}} \left[ \sum_{z^{rn} \in \{0,1\}^{rn}} \sum_{e^n \in \mathcal{B}_{pn}(0)} \mathbb{P}_{\mathcal{C}} \left( V(\mathcal{Z},z^{rn},e^n) > \delta 2^{(R-r)n} | \mathcal{H}^c(\mathcal{Z})\right) + \mathbb{P}_{\mathcal{C}}(\mathcal{H}(\mathcal{Z})) \right]
  \end{align}
  where we note that the $|\mathscr{Z}| \leq 2^n$ and $|\mathcal{B}_{pn}(0)| \leq 2^{nH_2(p)}$. Recall from Lemma \ref{thm:E_comp_1} that for $\mathcal{Z} \in \mathscr{Z}$, $\mathbb{P}_{\mathcal{C}}(\mathcal{H}(\mathcal{Z}))$ is at most $\alpha_k 2^{-(k\theta-r)n}$ for some constant $\alpha_k$ that depends only on $k$. For each $\mathcal{Z}$, $z^{rn}$ and $e^n$ in the sums of (\ref{eq:suff2_2}), we bound the conditional probability that $V(\mathcal{Z},z^{rn},e^n)$ is greater than $\delta 2^{(R-r)n}$.

  Let $\mathcal{Z} \in \mathscr{Z}$, $z^{rn} \in \{0,1\}^{rn}$ and $e^n \in \mathcal{B}_{pn}(0)$, and denote $\mathcal{H}(\mathcal{Z})$ as $\mathcal{H}$ without denoting the dependence on $\mathcal{Z}$. Applying Baye's formula again,
  \begin{align}
  &\mathbb{P}_{\mathcal{C}}\left(V(\mathcal{Z},z^{rn},e^n) > \delta 2^{(R-r)n} \big| \mathcal{H}^c\right)\nonumber \\
  & = \hspace{-1.5 em} \sum_{ \substack{ A \in  \{0,1\}^{ rn \times n}: \\ A \text{ consistent w/ } \mathcal{H}^c}} \hspace{-1em}\mathbb{P}_{\mathcal{C}}\left(V(\mathcal{Z},z^{rn},e^n) > \delta 2^{(R-r)n} \big| \mathcal{H}^c, \mathcal{Z}(G) = A \right)  \mathbb{P}_{\mathcal{C}}(\mathcal{Z}(G) = A|\mathcal{H}^c). \label{eq:suff2_3}
  \end{align}
  Suppose that for any $\mathcal{Z} \in \mathscr{Z}$, $z^{rn} \in \{0,1\}^{rn}$, $e^n \in \mathcal{B}_{pn}(0)$ and any $A$ consistent with $\mathcal{H}^c(\mathcal{Z})$, the quantity $\mathbb{P}_{\mathcal{C}}(V(\mathcal{Z},z^{rn},e^n) > \delta 2^{(R-r)n} | \mathcal{H}^c(\mathcal{Z}), \mathcal{Z}(G)=A)$ is at most $\beta(n)$. Then (\ref{eq:suff2_3}) is at most $\beta(n)$. In turn, (\ref{eq:suff2_2}) is at most $2^{n} \left( 2^{(r+H_2(p))n} \beta(n) + \alpha_k 2^{-(k\theta-r)} \right)$.
  \end{proof}

  \begin{remark} \label{rmk:drop_notation}
  In the sequel, we fix $\mathcal{Z} \in \mathscr{Z}$, $z^{rn} \in \{0,1\}^{rn}$, $e^n \in \mathcal{B}_{pn}(0)$ and $A \in \{0,1\}^{rn \times n}$ such that $A$ is consistent with $\mathcal{H}^c(\mathcal{Z})$. In turn, we drop the arguments from all functions/sets/variables denoting dependency on $e^n$ ,$\mathcal{Z}$ and $z^{rn}$ when the lack of notation does not cause confusion.
  \end{remark}

  \begin{remark} \label{thm:E_comp_dep}
  The indicator of the event that the $(n,Rn,k)$ code $\mathcal{C}$ is contained in $\mathcal{H}^c$ depends only on the rows of the generator matrix $G$ from the partition $\mathcal{Z}(G)$, and is independent of the rows from the partition $\mathcal{Z}^c(G)$. It follows that for any event $\mathcal{I}$, 
  \begin{equation} \nonumber
  \mathds{1}\{ \mathcal{H}^c \} \rightarrow \mathcal{Z}(G) \rightarrow \mathds{1} \{\mathcal{I}\}
  \end{equation}
  forms a Markov chain.
  \end{remark}

  The above remark can be verified by inspection of the definition of $\mathcal{H}$ (c.f. Definition \ref{def:H}). 
  
  The next result shows that we can condition on some information set of $\mathcal{C}$ while preserving the $k$-wise independence of codewords. We emphasize that beyond the $k$-wise independence property of $\mathcal{C}$, the linear structure of the pseudolinear code is critical here.

  \begin{lemma} \label{thm:E_comp_2}
  The view of the codewords of $\mathcal{C}$ at the adversary's unobserved coordinates $\mathcal{Z}^c$, i.e., $$\{ \mathcal{Z}^c(X(u)): u \in \mathcal{U} \},$$ conditioned on $\mathcal{H}^c$ and $\mathcal{Z}(G) = A$ are uniformly distributed in $\{0,1\}^{(1-r)n}$ and $k$-wise independent. 
  \end{lemma}

  \begin{proof}[Proof of Lemma \ref{thm:E_comp_2}] \label{sec:E_comp_2_proof}
  Let $u_1,\ldots,u_k$ be any $k$ unique messages in $\mathcal{U}$ and let $x^{(1-r)n}_1, \ldots, x^{(1-r)n}_k$ be any $k$ strings in $\{0,1\}^{(1-r)n}$. We show that 
  \begin{equation} \label{eq:E_comp_2}
  \begin{aligned}
  &\mathbb{P}_{\mathcal{C}}\left(\bigcap_{i=1}^k \left\{ \mathcal{Z}^c(X^n(u_i)) = x^{(1-r)n}_i \right\} \bigg| \mathcal{H}^c,\mathcal{Z}(G) = A\right) = \prod_{i=1}^k \mathbb{P}_{\mathcal{C}}\left( \mathcal{Z}^c(X^n(u_i)) = x^{(1-r)n}_i \right).
  \end{aligned}
  \end{equation}
  Let $\mathcal{F}_i$ denote the event ``$\mathcal{Z}^c(X^n(u_i)) = x^{(1-r)n}_i$''. Since $\mathds{1}\{\mathcal{H}^c\} \rightarrow \mathcal{Z}(G) \rightarrow \mathds{1}\{ \cap_{i=1}^k \mathcal{F}_i \}$ is a Markov chain, the LHS of (\ref{eq:E_comp_2}) is equal to
  \begin{align}
  \mathbb{P}_{\mathcal{C}} \left( \bigcap_{i=1}^k \mathcal{F}_i \bigg| \mathcal{Z}(G) = A \right) 
  \stackrel{(a)}{=} \mathbb{P}_{\mathcal{C}} \left( \bigcap_{i=1}^k \mathcal{F}_i \right) 
  \stackrel{(b)}{=} \prod_{i=1}^k \mathbb{P}_{\mathcal{C}} \left( \mathcal{F}_i \right) \nonumber
  \end{align}
  where (a) follows from the properties that i) $\mathcal{Z}^c(G)$ is independent of $\mathcal{Z}(G)$ and ii) $\mathcal{Z}^c(X^n(u_i))$ depends on $\mathcal{Z}^c(G)$ and is independent of $\mathcal{Z}(G)$, and (b) follows from the property that the codewords $X^n(u_1), \ldots, X^n(u_k)$ are independent (Lemma \ref{thm:k_wise_cw}).
  \end{proof}

  Recall that $V$ is defined in (\ref{eq:V}) as $V \triangleq \sum_{i=1}^S \sum_{u \in \mathcal{O} \cap \mathcal{S}_i \times \mathcal{S}^c_i} V(u,u')$. We break apart the sum $V$ into smaller sums which in turn are easier to analyze. For $i=1,2,\ldots,S$  define the sum 
  \begin{equation} \label{eq:Vi}
  V_i \triangleq \sum_{u \in \mathcal{O} \cap \mathcal{S}_i \times \mathcal{S}^c_i} V(u,u').
  \end{equation}
  Clearly, $V = \sum_{i=1}^S V_i$.

  \begin{lemma} \label{thm:cond_Vi}
  For $i \in [S]$ and conditioned on $\mathcal{H}$ and $\mathcal{Z}(G)=A$, the sum $V_i$ as defined in (\ref{eq:Vi}) has a conditional expectation 
  \begin{equation} \nonumber
  \mathbb{E}_{\mathcal{C}}\left[V_i|\mathcal{H}^c,\mathcal{Z}(G)=A\right] \leq n 2^{(R-r)n- \epsilon n + 2 \theta n}.
  \end{equation}
  \end{lemma}
  
  \begin{proof}[Proof of Lemma \ref{thm:cond_Vi}]
  Let $i \in [\mathcal{S}]$. We begin by characterizing the distribution of the terms $\{ V(u,u') \}$ of the sum $V_i$. Recall that $V(u,u')$ is defined as the indicator of the event ``$X^n(u') \in \mathcal{B}_{pn} \left( X^n(u) \oplus e^n\right)$''. Conditioned on the event ``$\mathcal{Z}(G)=A$'', the codeword of any message $u\in \mathcal{O} \cap \mathcal{S}_i$ over the coordinates $\mathcal{Z}$ is $\mathcal{Z}(X^n(u)) = z^{rn}$. Similarly, for any message $u' \in \mathcal{S}^c_i$, $\mathcal{Z}(X^n(u'))$ is deterministic. In turn, for any $u \in \mathcal{O} \cap \mathcal{S}_i$ and any $u' \in \mathcal{S}^c_i$, quantity $V(u,u')$ conditioned on $\mathcal{H}$ and $\mathcal{Z}(G)=A$ is the indicator of the event
  \begin{equation} \label{eq:event_ball}
   \mathcal{Z}^c(X^n(u')) \in \mathcal{B}_{\left( pn-D(u') \right)} \left( \mathcal{Z}^c(X^n(u)) \oplus \mathcal{Z}^c(e^n)\right)
  \end{equation}
  where $D(u') \triangleq d_H(\mathcal{Z}(X^n(u')),z^{rn} + \mathcal{Z}(e^n))$ is a deterministic quantity, and where $\mathcal{Z}^c(X^n(u))$ and $\mathcal{Z}^c(X^n(u'))$ are uniform over $\{0,1\}^{(1-r)n}$ and independent following Lemma \ref{thm:E_comp_2}. 
  
  With the above distribution in mind, we bound the conditional probability of event (\ref{eq:event_ball}). We use the well known fact that a Hamming ball of radius $t>0$ centered around any point in the space $\{0,1\}^{(1-r)n}$ has a volume at most $2^{(1-r)n H_2(\frac{t}{(1-r)n})}$ \cite{Macwilliams1977}. Hence, the conditional probability of event (\ref{eq:event_ball}) is
  \begin{align} 
  &\mathbb{E}_{\mathcal{C}}[V(u,u')|\mathcal{H},\mathcal{Z}(G)=A] \leq 2^{-(1-r)n \left(1 - H_2\left(\frac{pn - D(u')}{(1-r)n}\right)\right)} \label{eq:Vuu_cond}
  \end{align}
  if $D(u') \leq pn$ and $\mathbb{E}_{\mathcal{C}}[V(u,u')|\mathcal{H},\mathcal{Z}(G)=A] = 0$ if $D(u') > pn$.
  
  We now turn to the sum $V_i$. The conditional expectation of $V_i$ is
  \begin{align}
  &\mathbb{E}_{\mathcal{C}}[V_i\big|\mathcal{H},\mathcal{Z}(G)=A] = \hspace{-1em} \sum_{u \in \mathcal{O} \cap \mathcal{S}_i} \sum_{u' \in \mathcal{S}^c_i} \mathbb{E}_{\mathcal{C}}[V(u,u')\big|\mathcal{H},\mathcal{Z}(G)=A] \nonumber \\
  & \stackrel{(a)}{\leq} \sum_{u \in \mathcal{O} \cap \mathcal{S}_i} \sum_{\substack{u' \in \mathcal{S}^c_i: \\ D(u')\leq pn}} 2 ^{-(1-r)n \left(1 - H_2(\left( \frac{pn - D(u')}{(1-r)n}\right) \right)} \nonumber \\
  & \stackrel{(b)}{\leq} \sum_{\substack{u' \in \mathcal{S}^c_i: \\ D(u') \leq pn}} |\mathcal{O}| 2 ^{-(1-r)n \left(1 - H_2(\left( \frac{pn - D(u')}{(1-r)n}\right) \right)}. \label{eq:Vp_cond}
  \end{align}
  where (a) follows from (\ref{eq:Vuu_cond}) and (b) follows from the inequality $|\mathcal{O}\cap\mathcal{S}_i| \leq |\mathcal{O}|$. At this point, it is useful to recall some notation that we have thus far suppressed for presentation purposes (c.f. Remark \ref{rmk:drop_notation}). In particular, recall that the set $\mathcal{O} = \mathcal{O}(\mathcal{Z},z^{rn})$ depends on the adversary's coordinate set $\mathcal{Z}$ and observation string $z^{rn}$, both of which are fixed. With this notation in hand, we consider the dumby observation string $s^{rn} \in \{0,1\}^{rn}$ and partition $\mathcal{S}^c_i$ into subsets $\mathcal{S}^c_i \cap \mathcal{O}(\mathcal{Z},s^{rn})$ for all $s^{rn} \in \{0,1\}^{rn}$. For any message $u' \in \mathcal{S}^c_i \cap \mathcal{O}(\mathcal{Z},s^{rn})$, we have that $\mathcal{Z}(X^n(u')) = s^{rn}$ and we can write $D(u') = D(s^{rn}) = d_H(s^{rn},z^{rn} + \mathcal{Z}(e^n))$. In turn, (\ref{eq:Vp_cond}) is equal to 
  \begin{align}
  & \sum_{\substack{s^{rn} \in \{0,1 \}^{rn} \\ D(s^{rn}) \leq pn}} \sum_{u' \in \mathcal{S}^c_i \cap \mathcal{O}(\mathcal{Z},s^{rn})} |\mathcal{O}(\mathcal{Z},z^{rn})| 2^{-(1-r)n \left(1 - H_2\left(\frac{pn - D(s^{ rn})}{(1-r)n}\right)\right)} \nonumber \\
  & \stackrel{(c)}{\leq} \sum_{\substack{s^{rn} \in \{0,1\}^{rn}: \\ D(s^{rn}) \leq pn}} 2^{2(R-r)n +2 \theta n -(1-r)n \left(1 - H_2\left(\frac{pn - D(s^{rn})}{(1-r)n}\right) \right) } \nonumber \\
  & = \sum_{j=0}^{\min \{pn,rn\}} {rn \choose j} 2^{2(R-r)n +2 \theta n -(1-r)n \left(1 - H_2\left(\frac{pn - j}{(1-r)n}\right) \right) } \nonumber \\
  & \stackrel{(d)}{\leq} \sum_{j=0}^{\min \{pn,rn\}} 2^{rnH_2(\frac{j}{rn})+2(R-r)n +2 \theta n -(1-r)n \left(1 - H_2\left(\frac{pn - j}{(1-r)n}\right) \right) }. \nonumber
  \end{align}
  where (c) follows from the definitional property that given $\mathcal{H}^c$, $|\mathcal{O}(\mathcal{Z},s^{rn})| \leq 2^{(R-r)n+\theta n}$ for all $s^{rn} \in \{0,1\}^{rn}$, and (d) follows from the fact that the quantity ${rn \choose j}$ is bounded above by $2^{rn H_2(\frac{j}{rn})}$.

  To finish the proof, we show that the exponent 
  \begin{align}
  &E(j) = rn H_2\left(\frac{j}{rn}\right) + 2(R-r)n + 2 \theta n  - (1-r)n \left( 1 - H_2\left( \frac{pn-j}{(1-r)n}\right) \right) \nonumber
  \end{align}
  of the term inside the sum of (d) is bounded above by $E(rpn) = (R-r)n - \epsilon n + 2\theta n$ for all $j \in \{0,1,\ldots,\min \{pn,rn\} \}$. We remark that for the function $E(x)$ over a continuous variable $x \in [0,\min\{pn,rn\}]$, the derivative $\frac{\partial E(rpn)}{\partial x}$ is $0$. Thus, it is sufficient to show that $E(x)$ is concave over $x \in [0,\min\{pn,rn\}]$. Towards this end, note that the terms $H_2(\frac{x}{rn})$ and $H_2(\frac{pn-x}{(1-r)n})$ are concave following that they are affine compositions of the concave function $H_2(x)$. It follows that $E(x)$ is a positive-weighted sum of concave functions, which in turn, is concave.
  \end{proof}

  \begin{lemma} \label{thm:Vk_wise_forest}
  For $i \in [S]$ and conditioned on $\mathcal{H}^c$ and $\mathcal{Z}(G) = A$, the random variables $\{V(u,u')\}_{(u,u') \in \mathcal{O} \cap \mathcal{S}^c_i}$ are $\lfloor k/2 \rfloor$-wise independent over every forest.
  \end{lemma}
  \begin{proof}[Proof of Lemma \ref{thm:Vk_wise_forest}]
  Let $i \in [S]$ and let $\mathcal{F} = (\mathcal{O} \cap \mathcal{S}_i,\mathcal{S}^c_i,\mathcal{E})$ be an undirected bipartite forest with vertex parts $\mathcal{O} \cap \mathcal{S}_i$ and $\mathcal{S}^c_i$, and edge set $\mathcal{E}$ of size $\ell \triangleq \lfloor k/2 \rfloor$. We show that the random variables $\{V(u,u')\}_{(u,u') \in \mathcal{E}}$ are conditionally $\ell$-wise independent conditioned on $\mathcal{H}^c$ and $\mathcal{Z}(G)=A$.

  Let $(u_1,u'_1), (u_2,u'_2), \ldots, (u_\ell,u'_\ell)$ be an enumeration of all edges in $\mathcal{F}$ such that at least one of the vertices in $\{u_j,u_j'\}$ is a leaf of a tree in $\mathcal{F} \setminus \{(u_{j+1},u'_{j+1}), \ldots, (u_{\ell},u'_{\ell}) \}$. We remark that such an enumeration exists since $\mathcal{F}$ is acyclic. The following sufficient condition follows from Baye's formula: $\{V(u,u')\}_{(u,u') \in \mathcal{E}}$ are conditionally $\lfloor k/2 \rfloor$-wise independent conditioned on $\mathcal{H}^c$ and $\mathcal{Z}(G)=A$ if for all $j \in [\ell]\setminus \{1\}$ the random variable $V(u_j,u'_j)$ is independent of $\{V(u_{j'},u'_{j'})\}_{j' = 1}^{j-1}$ conditioned on $\mathcal{H}^c$ and $\mathcal{Z}(G)=A$.\footnote{Independence follows from Baye's formula and holds for any enumeration.} 

  Let $j\in [\ell] \setminus \{1\}$. In the following, suppose that we condition on $\mathcal{H}^c$ and $\mathcal{Z}(G)=A$. Recall that the codeword views $\mathcal{Z}(X^n(u_{j'}))$ and $\mathcal{Z}(X^n(u'_{j'}))$ are deterministic for all $j' = 1,\ldots,j$. Thus, the random variable $V(u_j,u'_j)$ is equal to the indicator of the event ``$\mathcal{Z}^c(X^n(u'_j)) \in \mathcal{B}_{(pn-D)}\left( \mathcal{Z}^c(X^n(u_j) \oplus e^n) \right)$'' where $D \triangleq d_H(\mathcal{Z}(X^n(u'_j)),\mathcal{Z}(X^n(u_j)) \oplus e^n)$ is a deterministic quantity. By our choice of enumeration, either $\mathcal{Z}^c(X^n(u_j))$ or $\mathcal{Z}^c(X^n(u'_j))$ or (both) are not in the set of random variables $\{ \mathcal{Z}(X^n(u_{j'})) \}_{j'=1}^{j-1} \cup \{ \mathcal{Z}(X^n(u'_{j'})) \}_{j'=1}^{j-1}$ of which the set $\{V(u_{j'},u'_{j'})\}_{j' = 1}^{j-1}$ is a deterministic mapping. The desired independence follows from the property that $\{ \mathcal{Z}(X^n(u_{j'})) \}_{j'=1}^{j-1} \cup \{ \mathcal{Z}(X^n(u'_{j'})) \}_{j'=1}^{j}$ are $k$-wise independent (Lemma \ref{thm:E_comp_2}).   
  \end{proof}

  The following result together with the second sufficient condition (Lemma \ref{thm:suff2}) completes the proof of Theorem \ref{thm:reliable}.
  
  \begin{lemma} \label{thm:V_conc}
  Set $\theta = \epsilon/4$. Then $\mathbb{P}_{\mathcal{C}}\left( V > \delta 2^{(R-r)n} \big| \mathcal{H}^c,\mathcal{Z}(G) = A \right)$ is bounded above by $$\gamma_k \frac{(10n)^{1 + \lfloor k/2 \rfloor}}{\delta} 2^{-\lfloor k/2 \rfloor\frac{\epsilon}{2}n}$$ for some constant $\gamma_k$ that depends only on $k$.
  \end{lemma}

  \begin{proof}[Proof of Lemma \ref{thm:V_conc}]
   
  Recall that $V$ is by definition equal to $\sum_{i=1}^S V_i$ where $V_i$ is defined in (\ref{eq:Vi}). Consider the probability that $V$ is greater than $\delta 2^{(R-r)n}$ conditioned on $\mathcal{H}^c$ and $\mathcal{Z}(G) = A$. Following a simple union bound, this probability is bounded above by 
  \begin{equation} \label{eq:V_conc_1}
  \sum_{i=1}^{S} \mathbb{P}_{\mathcal{C}}\left( V_i > \frac{\delta 2^{(R-r)n}}{S} \bigg| \mathcal{H}^c,\mathcal{Z}(G) = A \right).
  \end{equation}

  Our goal is to bound (\ref{eq:V_conc_1}) by applying the concentration inequality of Lemma \ref{thm:conc_ineq_2}. For $i \in [S]$, we recall the following properties of the random variable $V_i$:
  \begin{enumerate}
  \item $V_i$ is a sum over $M_1 M_2$ terms where $M_1 \triangleq |\mathcal{O} \cap \mathcal{S}_i|$ and $M_2 \triangleq |\mathcal{S}_i^c|$.
  \item The conditional expectation $\mu_i \triangleq \mathbb{E}_{\mathcal{C}}\left[ V_i |\mathcal{H}^c,\mathcal{Z}(G)=A\right]$ is bounded above by $n 2^{(R-r)n- n \epsilon + 2 \theta n}$ (Lemma \ref{thm:cond_Vi}).
  \item Conditioned on $\mathcal{H}^c$ and $\mathcal{Z}(G)=A$, the random variables $\{V(u,u')\}_{(u,u') \in \mathcal{O} \cap \mathcal{S}_i \times \mathcal{S}^c_i}$ are $\lfloor k/2 \rfloor$-wise independent over every forest (Lemma \ref{thm:Vk_wise_forest}).
  \end{enumerate}
  Then Lemma \ref{thm:conc_ineq_2} implies that (\ref{eq:V_conc_1}) is bounded above by
  \begin{align} \label{eq:V_ub_2}
  \gamma_k' S  {M_1 M_2 \choose \lfloor k/2 \rfloor} \left( \frac{2^{(R-r)n - \epsilon n + 2 \theta n}}{M_1 M_2} \right)^{\lfloor k/2 \rfloor} { \frac{\delta 2^{(R-r)n}}{S} \choose \lfloor k/2 \rfloor}^{-1}
  \end{align}
  where $\gamma_k'$ is a constant that depends only on $k$. We can simplify the above expression by using the well known bounds $n^k/k! \leq {n \choose k} \leq n^k/k^k$, where it follows that (\ref{eq:V_ub_2}) is bounded above by $$\gamma_k \frac{(10n)^{1+ \lfloor k/2 \rfloor}}{\delta} 2^{-\lfloor k/2 \rfloor(\epsilon - 2 \theta)n}$$ for a constant $\gamma_k$ that depends only on $k$. Setting $\theta = \epsilon/4$ gives the desired result.
  \end{proof}

  \appendices

  \section{Proof of Lemma \ref{thm:conc_ineq_2}} \label{sec:conc_ineq_2_proof}

  Let $\mathcal{V}_1$ and $\mathcal{V}_2$ be disjoint sets of integers of size $M_1$ and $M_2$, respectively. Let $\{V_{i,j}\}_{(i,j) \in \mathcal{V}_1 \times \mathcal{V}_2}$ be a set of binary random variables that take values in $\{0,1\}$. Furthermore, define the sum $V \triangleq \sum_{i \in \mathcal{V}_1} \sum_{j \in \mathcal{V}_2} V_{i,j}$ and define the sum mean $\mu \triangleq \mathbb{E}[V] = \sum_{i \in \mathcal{V}_1} \sum_{j \in \mathcal{V}_2} \mu_{i,j}$ where $\mu_{i,j} \triangleq \mathbb{E}[V_{i,j}]$. Suppose that $\mu \geq 1$. Furthermore, let $k \geq 2$ and suppose that $\{V_{i,j}\}_{(i,j) \in \mathcal{V}_1 \times \mathcal{V}_2}$ are $k$-wise independent over every forest. 
  
  We begin by restating a general concentration inequality that is not specific to the joint distribution of $\{V_{i,j}\}$ and only uses the fact that $\{V_{i,j}\}$ are binary in $\{0,1\}$.

  \begin{claim}[{Schmidt, Siegel and Srinivasan \cite{Schmidt1995}}] \label{thm:claim_SSS} For any $\lambda > 0$,
  \begin{equation} \label{eq:kw_b1}
  \mathbb{P}\left(V > \lambda \right) \leq \frac{1}{{\lambda \choose k}} \sum\limits_{\mathcal{E} \in \mathcal{P}^{(k)}} \mathbb{E}\left[ \prod\limits_{(i,j) \in \mathcal{E}} V_{i,j} \right]
  \end{equation}
  where $\mathcal{P}^{(k)} \triangleq  \{ \mathcal{E} \subseteq \mathcal{V}_1 \times \mathcal{V}_2: |\mathcal{E}| = k \}$ and where we emphasize that $|\mathcal{E}|$ denotes the number of edges $(i,j) \in \mathcal{V}_1 \times \mathcal{V}_2$ in $\mathcal{E}$.
  \end{claim}

  \begin{proof}[Proof of Claim \ref{thm:claim_SSS}]
  For completeness, we restate the proof from \cite{Schmidt1995}. Consider the symmetric function
  \begin{equation} \nonumber
  Y_k\left(\{V_{i,j}\}\right) =  \sum_{\mathcal{E} \in \mathcal{P}^{(k)}} \prod_{(i,j) \in \mathcal{E}} V_{i,j}.
  \end{equation}
  Observe that since $\{V_{i,j}\}$ are binary variables in $\{0,1\}$, we have that for any integer $\ell \geq 0$, $\sum_{(i,j) \in \mathcal{V}_1 \times \mathcal{V}_2} V_{i,j} = \ell$ if and only if $Y_k(\{V_{i,j}\}) = {\ell \choose k}$. In turn,
  \begin{equation} \nonumber
  \mathbb{P} \left( V > \lambda \right) = \mathbb{P} \left( Y_k\left( \{V_{i,j}\} \right) > {\lambda \choose k} \right) \leq \frac{1}{{\lambda \choose k}} \sum_{\mathcal{E} \in \mathcal{P}^{(k)}} \mathbb{E} \left[ \prod_{(i,j)\in \mathcal{E}} V_{i,j} \right]
  \end{equation}
  where the inequality follows from Markov's inequality.
  \end{proof}

  The goal of the proof is to bound the RHS of (\ref{eq:kw_b1}) using the limited independence properties of $\{V_{i,j}\}$. The following setup is needed. For a given subset of vertex pairs $\mathcal{E} \subseteq \mathcal{V}_1 \times \mathcal{V}_2$ of size $|\mathcal{E}| = k$, we define an undirected bipartite graph $\mathcal{G} = (\mathcal{V}_1,\mathcal{V}_2,\mathcal{E})$ with vertex parts $\mathcal{V}_1$ and $\mathcal{V}_2$ and edge set $\mathcal{E}$. We emphasize that $\mathcal{G}$ depends on edge set $\mathcal{E}$ by writing $\mathcal{G}(\mathcal{E})$. 

  We define the following graph theoretic concepts that will be use throughout the proof. Two vertices (edges) in $\mathcal{G}(\mathcal{E})$ are said to be \textit{adjacent} if they share an edge (vertex). A \textit{cycle} of a graph $\mathcal{G}(\mathcal{E})$ is a sequence of adjacent vertices in which only the first and last vertex in the sequence are equal. A \textit{maximum spanning forest} $\mathcal{F}(\mathcal{E})$ of a graph $\mathcal{G}(\mathcal{E})$ is a forest subgraph of $\mathcal{G}(\mathcal{E})$ with the additional property that a cycle will necessarily be created in $\mathcal{F}(\mathcal{E})$ if we add any edge $e$ from $\mathcal{G}(\mathcal{E})$ to $\mathcal{F}(\mathcal{E})$ such that $e$ is not in $\mathcal{F}(\mathcal{E})$. Note that for $\mathcal{E} \in \mathcal{P}^{(k)}$ the maximum spanning forest of $\mathcal{G}(\mathcal{E})$ is not necessarily unique. A \textit{fundamental cycle} of $\mathcal{G}(\mathcal{E})$ with respect to a maximum spanning forest $\mathcal{F}(\mathcal{E})$ of $\mathcal{G}(\mathcal{E})$ is a cycle that can be created by taking an edge $e$ of $\mathcal{G}(\mathcal{E})$ that is not in $\mathcal{F}(\mathcal{E})$ and adding $e$ to $\mathcal{F}(\mathcal{E})$. We note that the number of fundamental cycles of $\mathcal{G}(\mathcal{E})$ does not depend on $\mathcal{F}(\mathcal{E})$, and thus, we can talk about the number of fundamental cycles without specifying a maximum spanning forest. 

  For the set $\mathcal{P}^{(k)} \triangleq \{\mathcal{E} \subseteq \mathcal{V}_1 \times \mathcal{V}_2: |\mathcal{E}| = k\}$ and for $s\geq 0$, define 
  \begin{equation} \nonumber
  \mathcal{P}^{(k)}_s = \left\{ \mathcal{E} \in \mathcal{P}^{(k)}: \mathcal{G}(\mathcal{E}) \text{ has exactly $s$ fundamental cycles} \right\}.
  \end{equation}
  Let $S \geq 0$ be the integer such that $\mathcal{P}^{(k)}_S$ is non-empty and $\mathcal{P}^{(k)}_{S+1}$ is empty; note that $S$ is strictly less than the number of edges $k$. Hence, the subsets $\mathcal{P}^{(k)}_0, \mathcal{P}^{(k)}_1, \ldots, \mathcal{P}^{(k)}_{S}$ are a partition of $\mathcal{P}^{(k)}$, and in turn, we rewrite the RHS of (\ref{eq:kw_b1}) as
  \begin{equation} \label{eq:kw_b2}
  \frac{1}{{\lambda \choose k}} \sum_{s = 0}^{S} \sum_{\mathcal{E}\in\mathcal{P}^{(k)}_s} \mathbb{E} \left[ \prod_{(i,j)\in \mathcal{E}} V_{i,j} \right].
  \end{equation}
  For each edge set $\mathcal{E} \in \mathcal{P}^{(k)}$, we fix $\mathcal{F}(\mathcal{E})$ to be any maximum spanning forest of $\mathcal{G}(\mathcal{E})$. For an edge set $\mathcal{E} \in \mathcal{P}^{(k)}$ let the notation $(i,j) \in \mathcal{F}(\mathcal{E})$ denote an edge $(i,j)$ in the maximum spanning forest $\mathcal{F}(\mathcal{E})$. Following that the edge set of $\mathcal{F}(\mathcal{E})$ is contained in $\mathcal{E}$ and together with the trivial bound $V_{i,j} \leq 1$, (\ref{eq:kw_b2}) is bounded above by
  \begin{align}
  & \frac{1}{{\lambda \choose k}} \sum_{s = 0}^{S} \sum_{\mathcal{E}\in\mathcal{P}^{(k)}_s} \mathbb{E} \left[ \prod_{(i,j)\in \mathcal{F}(\mathcal{E})} V_{i,j} \right] \nonumber \\
  & \stackrel{(\ast)}{=} \frac{1}{{\lambda \choose k}} \sum_{s = 0}^{S} \sum_{\mathcal{E}\in\mathcal{P}^{(k)}_s}  \prod_{(i,j)\in \mathcal{F}(\mathcal{E})} \mathbb{E} \left[ V_{i,j} \right] \nonumber \\
  & = \frac{1}{{\lambda \choose k}} \sum_{s = 0}^{S} \sum_{\mathcal{E}\in\mathcal{P}^{(k)}_s}  \prod_{(i,j)\in \mathcal{F}(\mathcal{E})} \mu_{i,j} \label{eq:kw_b4}
  \end{align}
  where ($\ast$) follows from the fact that the random variables $\{V_{i,j}\}_{(i,j) \in \mathcal{V}_1 \times \mathcal{V}_2}$ are $k$-wise independent over every forest. While (\ref{eq:kw_b4}) is a function of terms $\{ \mu_{i,j} \}$, we would like a bound that depends on the sum mean $\mu$ and not on $\{\mu_{i,j}\}$.

  Towards this desired bound, we make the following observation: for $s \in [S] \cup \{0\}$ and for any forest $\mathcal{F}' = (\mathcal{V}_1,\mathcal{V}_2,\cdot)$ with $k-s$ edges, there exists at most a constant number (say $b_{k} \geq 1$) depending only on $k$ of edge sets $\mathcal{E} \in \mathcal{P}^{(k)}_s$ such that $\mathcal{F}(\mathcal{E}) = \mathcal{F}'$. This observation follows from the fact that there are at most a constant number of ways to choose an edge set $\mathcal{E} \in \mathcal{P}^{(k)}_s$ by assigning $s$ free edges to be incident to the fixed $k-s$ non-isolated vertices of forest $\mathcal{F}'$. It follows that for $s \in [S] \cup \{0\}$ and for every edge set $\mathcal{E}' \in \mathcal{P}^{(k-s)}$ of size $k-s$, there are at most $b_k$ number of edge sets $\mathcal{E} \in \mathcal{P}^{(k)}_s$ such that $\mathcal{E}' = \mathcal{F}(\mathcal{E})$. In turn, 
  \begin{equation} \nonumber 
  \sum_{\mathcal{E} \in \mathcal{P}^{(k)}_s} \prod_{(i,j) \in \mathcal{F}(\mathcal{E})} \mu_{i,j} \leq b_k \sum_{\mathcal{E}' \in \mathcal{P}^{(k-s)}} \prod_{(i,j) \in \mathcal{E}'} \mu_{i,j} \triangleq g\left( \{\mu_{i,j}\} \right).
  \end{equation}
  Now, using the same approach as \cite[Lemma 2]{Schmidt1995}, we maximize the function $g(\{\mu_{i,j}\})$ with respect to the parameters $\{\mu_{i,j}\}$ subject to the constraints $\sum_{i \in \mathcal{V}_1} \sum_{j \in \mathcal{V}_2} \mu_{i,j} = \mu$ and $0 \leq \mu_{i,j} \leq 1$ for all $(i,j)\in \mathcal{V}_1 \times \mathcal{V}_2$. We remark that one can show that the assignment $\mu_{i,j} = \frac{\mu}{M_1 M_2}$ for all $(i,j) \in \mathcal{V}_1 \times \mathcal{V}_2$ is a global maximum of $g(\left\{ \mu_{i,j} \right\})$. To see this, suppose that for a set of parameters $\{\mu_{i,j}\}$ there exists a $\mu_{i_1,j_1}$ and $\mu_{i_2,j_2}$ such that $\mu_{i_1,j_1} < \frac{\mu}{M_1 M_2} < \mu_{i_2,j_2}$. Then for any small number $\eta > 0$ such that $\mu_{i_1,j_1} + \eta \leq \frac{\mu}{M_1 M_2} \leq \mu_{i_2,j_2} - \eta$ and for the updated parameters $\{\mu'_{i,j}\}$ where 
  \begin{align}
  &\mu'_{i_1,j_1} = \mu_{i_1,j_1} + \eta, \nonumber \\
  &\mu'_{i_2,j_2} = \mu_{i_2,j_2} - \eta, \nonumber \\
  &\mu'_{i,j} = \mu_{i,j} \text{ for all } (i,j) \in \mathcal{V}_1 \times \mathcal{V}_2 \setminus \{(i_1,j_1), (i_2,j_2) \}, \nonumber
  \end{align}
  it is straightforward to show that $g(\{\mu_{i,j}\})$ is bounded above by $g(\{ \mu'_{i,j}\})$. In summary, $g(\{\mu_{i,j}\})$ is bounded such that
  \begin{equation} \label{eq:last_bd_g}
  g(\{\mu_{i,j}\}) \leq |\mathcal{P}^{(k-s)}| \left( \frac{\mu}{M_1 M_2} \right)^{k-s} = { M_1 M_2 \choose k-s} \left( \frac{\mu}{M_1 M_2} \right)^{k-s}.
  \end{equation}

  We now return our focus to the quantity (\ref{eq:kw_b4}). Applying the bound (\ref{eq:last_bd_g}), (\ref{eq:kw_b4}) is bounded above by 
  \begin{align}
  & \frac{b_k}{{\lambda \choose k}} \sum_{s= 0}^{S} |\mathcal{P}^{(k-s)}| \left( \frac{\mu}{M_1 M_2}\right)^{k-s} \nonumber \\
  & = \frac{b_k}{{\lambda \choose k}} \sum_{s= 0}^{S} {M_1 M_2 \choose k-s} \left( \frac{\mu}{M_1 M_2}\right)^{k-s} \nonumber
  \end{align}
  which in turn, following the fact that $S \leq k$ and the fact that ${M_1 M_2 \choose k-s} \left( \frac{\mu}{M_1 M_2} \right)^{k-s} = \Theta_k(\mu^{k-s})$ where the $k$ subscript indicates that constant hidden by the big Theta notation depends only on $k$,\footnote{This fact follows from the well known inequalities $\frac{n^k}{k^k} \leq {n \choose k} \leq \frac{n^k}{k!}$.} is bounded above by
  \begin{equation} \nonumber
  \Theta_k \left( \frac{1}{{\lambda \choose k}} {M_1 M_2 \choose k} \left( \frac{\mu}{M_1 M_2}\right)^{k} \right).
  \end{equation}
  Here, we have used the assumption that $\mu \geq 1$, which allows us to bound $\Theta_k(\mu^{k-s})$ above by $\Theta_k(\mu^k)$. This completes the proof of Lemma \ref{thm:conc_ineq_2}.

  \section{Proof of Lemma \ref{thm:k_wise_cw}} \label{sec:k_wise_cw_proof}
 
  Suppose that $\mathcal{C}$ is an $(n,Rn)$ code drawn from distribution $F(n,Rn,k)$. Let $u_1, \ldots, u_k$ be $k$ unique messages in $\mathcal{U}$ and let $x^n_1, \ldots, x^n_k$ be strings in $\{0,1\}^n$. The probability that the $k$ codewords in $\mathcal{C}$ corresponding to messages $u_1,\ldots,u_k$ coincide with the strings $x^n_1,\ldots,x^n_k$, respectively, is
  \begin{equation} \label{eq:k_wise_cw_1}
  \mathbb{P}_{\mathcal{C}}\left(\bigcap_{i=1}^k \{X^n(u_i) = x^n_i\} \right) = \mathbb{P}_{\mathcal{C}}\left(G H_{(u_1,\ldots,u_k)} = A \right)
  \end{equation}
  where $H_{(u_1,\ldots,u_k)}$ denotes the matrix $\begin{bmatrix} h^{m}(u_1) & \cdots & h^{m}(u_k) \end{bmatrix}$ and $A$ is the matrix $\begin{bmatrix} x^n_1 & \cdots & x^n_k \end{bmatrix}$. Via (\ref{eq:k_wise_cw_1}), it is clear that the codewords $X^n(u_1), \ldots, X^n(u_k)$ are uniformly distributed in $\{0,1\}^n$ and independent if the columns of matrix $G H_{(u_1,\ldots,u_k)}$ are identically distributed in $\{0,1\}^n$ and independent.  
  
  Recall that the binary check matrix $H$ belongs to a linear code with minimum distance at least $k+1$. It follows that the columns $h^{m}(u_1), \ldots, h^{m}(u_k)$ are linearly independent (see, e.g., \cite{Guruswami2019}), and in turn, the linear equation $G H_{(u_1,\ldots,u_k)} = A$ is equivalent to the equation $G H'_{(u_1,\ldots,u_k)} = A'$ 
  where $H'_{(u_1,\ldots,u_k)}$ is some binary matrix with dimension equal to that of $H_{(u_1,\ldots,u_k)}$ such that there are $k$ rows of $H'_{(u_1,\ldots,u_k)}$ (denote their indices $j_1,\ldots,j_k$) that coincide with the rows of the identity matrix $I_k$ of dimension $k$, and where $A'$ is some binary matrix with dimension equal to that of $A$. The structure of $H'_{(u_1,\ldots,u_k)}$ together with the fact that the columns of $G$ are i.i.d. uniform implies that matrix $G H'_{(u_1,\ldots,u_k)}$ is equal in distribution to $\begin{bmatrix} g^{n}_{j_1} \cdots g^{n}_{j_k} \end{bmatrix}$, and thus, the columns of $H'_{(u_1,\ldots,u_k)}G$ are uniform in $\{0,1\}^n$ and independent.

\bibliographystyle{IEEEtran}
\bibliography{refs}

%
\IEEEpeerreviewmaketitle

\end{document}